%% file: UniformSampling.tex
\documentclass[11pt]{article}
\usepackage{amsmath, amsthm, amssymb}
\usepackage{color}
\usepackage{fullpage}

\usepackage[noend]{algorithmic}
\usepackage{algorithm}

\usepackage{enumerate}

\usepackage{hyperref}

\usepackage{graphicx}
\usepackage{caption}
\usepackage{subcaption}

\newcommand{\E}{\operatorname{E}}

\newcommand*\samethanks[1][\value{footnote}]{\footnotemark[#1]}
 \newtheorem{theorem}{Theorem}[section]
 \newtheorem{lemma}[theorem]{Lemma}

 \newtheorem{fact}[theorem]{Fact}
 
 \newtheorem{definition}[theorem]{Definition}

\newenvironment{proofof}[1]{\begin{proof}[proof of {#1}]}{\end{proof}}

\title{Applications of Uniform Sampling:\\ Densest Subgraph and Beyond\footnote{This paper is available on arxiv since June 15, 2015}}

\author{
Hossein Esfandiari\thanks{Supported in part by NSF CAREER award 1053605, NSF grant CCF-1161626, ONR YIP award N000141110662, DARPA/AFOSR grant FA9550-12-1-0423.}\\University of Maryland  
\and MohammadTaghi Hajiaghayi\samethanks[2]\\University of Maryland
\and David P. Woodruff\\IBM Almaden
}

\begin{document}
\sloppy
\maketitle
\setcounter{page}{0}
\thispagestyle{empty}

\input{abstract.tex}

\newpage

\section{Introduction}
\input{intro.tex}

\section{Densest Subgraph}\label{sec:DS}
\input{DensestSubgraph.tex}

\section{A General Family of Problems}\label{sec:GF}
\input{GeneralFamily.tex}

\section{Applications}\label{sec:App}
\input{Applications.tex}

\bibliographystyle{plain}
\bibliography{densest}



\end{document}

%% file: abstract.tex
\begin{abstract}
Recently [Bhattacharya et al., STOC 2015] \cite{bhattacharya2015space} provide the first non-trivial algorithm for the \emph{densest subgraph problem} in the streaming model with additions and deletions to its edges, i.e., for dynamic graph streams. They present a $(0.5-\epsilon)$-approximation algorithm using $\tilde{O}(n)$ space, where factors of $\epsilon$ and $\log(n)$ are suppressed in the $\tilde{O}$ notation. However, the update time of this algorithm is large. To remedy this, they also provide a $(0.25-\epsilon)$-approximation algorithm using $\tilde{O}(n)$ space with update time $\tilde{O}(1)$.

In this paper we improve the algorithms in \cite{bhattacharya2015space} by providing a $(1-\epsilon)$-approximation algorithm using $\tilde{O}(n)$ space. 
Our algorithm is conceptually simple - it samples $\tilde{O}(n)$ edges uniformly at random, and finds the densest subgraph on the sampled graph. We also show how to perform this sampling with update time $\tilde{O}(1)$.
In addition to this, we show that given oracle access to the edge set, we can implement our algorithm in time $\tilde{O}(n)$ on a graph in the standard RAM model. To the best of our knowledge this is the fastest $(0.5-\epsilon)$-approximation algorithm for the densest subgraph problem in the RAM model given such oracle access. 

Interestingly, we extend our results to a general class of graph optimization problems that we call \emph{heavy subgraph problems}. This class contains many interesting problems such as densest subgraph, directed densest subgraph, densest bipartite subgraph, $d$-max cut, and $d$-sum-max clustering.  
Our results, by characterizing heavy subgraph problems, address Open Problem 13 at the IITK Workshop on Algorithms for Data Streams in 2006 regarding the effects of subsampling, in the context of graph streams. 
\end{abstract}

%% file: intro.tex
In this paper we consider a general class of graph optimization problems that we call \emph{heavy subgraph problems} in the streaming setting with additions and deletions, i.e., in dynamic graph streams. We show that many interesting problems such as densest subgraph, directed densest subgraph, densest bipartite subgraph, $d$-max cut, and $d$-sum-max clustering fit in this general class of problems. To the best of our knowledge, we are the first to consider densest bipartite subgraph and $d$-sum-max clustering in the streaming setting. We defer the definitions of these two problems to Subsections \ref{subsec:DBS} and \ref{subsec:DSMC}, respectively. 

Finding the densest subgraph of a graph is one of the classical problems in computer science and appears to have many applications that deal with massive data such as spam detection \cite{gibson2005discovering} and analyzing communication in social networks \cite{dourisboure2007extraction}. 

In an instance of the \emph{densest subgraph problem} we are given a graph $G$ and want to find a subgraph $sol\subseteq G$, that maximizes $\frac {|E_{sol}|}{|V_{sol}|}$, where $V_{sol}$ and $E_{sol}$ are the vertex set and the edge set of $sol$, respectively.
Similarly, in an instance of the \emph{directed densest subgraph problem} we are given a directed graph $G=(V_G,E_G)$ and want to find a pair $A, B\subseteq V_G$, that maximizes $\frac {|E(A,B)|}{\sqrt{|A|\cdot|B|}}$, where $E(A,B)$ is the number of edges from $A$ to $B$ in $E_G$.

Charikar \cite{charikar2000greedy} studies the densest subgraph problem in the classical setting and provides a $0.5$-approximation algorithm with running time $O(n+m)$, where $n$ and $m$ are the number of vertices and edges, respectively. To the best of our knowledge, this is the fastest known constant approximation algorithm for the densest subgraph problem. Moreover, he provides a $0.5$-approximation algorithm for the directed densest subgraph problem.

Later, Bahmani, Kumar and Vassilvitskii \cite{bahmani2012densest}, consider the densest subgraph problem and the directed densest subgraph problem in the streaming setting with only insertions of edges. For both problems, they present streaming algorithms with a $\frac{0.5}{1+\epsilon}$-approximation factor using $\log_{1+\epsilon}(n)$ passes over the input. To the best of our knowledge, their results for directed graphs are the only non-trivial results for the directed densest subgraph problem in the streaming setting, prior to our work.

Recently, Bhattacharya et al. \cite{bhattacharya2015space} present the first single pass streaming algorithm for dynamic graph streams. Their first algorithm provides a $(0.5-\epsilon)$-approximation using $\tilde{O}(n)$ bits of space, though the update time is inefficient. They provide a second algorithm with a $(0.25-\epsilon)$-approximation factor for which the update time is only $\tilde{O}(1)$, again using $\tilde{O}(n)$ bits of space. 

In the \emph{$d$-max cut} problem, we are given a graph $G$, and we want to decompose the vertices of $G$ into $d$ partitions such that the number of edges between different partitions is maximized. To the best of our knowledge, we are the first to consider this problem for general $d$ in the streaming setting. A restricted version of this problem where $d=2$, is the classic max cut problem. One can store a sparsifier \cite{ahn2012graph} of the input graph in $\tilde{O}(n)$ space, and preserve a $(1-\epsilon)$-approximation of the max cut. However, it is not clear if sparsifiers preserve $d$-max cut or not. 
Recently, Kapralov, Khanna and Sudan \cite{kapralov2015streaming} show that any $(1-\epsilon)$-approximation to the max cut problem in the streaming setting requires $n^{1-O(\epsilon)}$ space.

 In this paper, we refer to $\frac {|E_{sol}|}{|V_{sol}|}$ as the density of the solution $sol$, and denote it by $den(sol)$. We refer to the density of the densest subgraph of $G$ by $opt(G)$. We also denote the densest subgraph among all those on $k$ vertices by $opt(G,k)$. We sometimes abuse notation and use $opt(G)$ to refer to both the densest subgraph as well as its density. 
 It is easy to see that the densest subgraph of a graph $G$ is an induced subgraph of it, otherwise we can simply add more edges to it. Therefore, we can indicate the densest subgraph by its vertex set. 
 For a graph $G$, and a subset of vertices $U$, we denote the induced subgraph of $G$ on vertex set $U$ by $G[U]$.

\subsection{Our Results}

In this paper, we first consider the densest subgraph problem in the streaming setting where we have both insertions and deletions to the edges as they arrive in the stream, i.e., in a dynamic graph stream. 
We improve the result of Bhattacharya et al. (STOC'15) \cite{bhattacharya2015space} by providing a $(1-\epsilon)$-approximation algorithm for this problem using $\tilde{O}(n)$ space.
Indeed, our algorithm simply samples $\tilde{O}(n)$ edges uniformly at random, and finds the densest subgraph on the sampled graph. We also achieve update time $\tilde{O}(1)$. To achieve this, we use min-wise independent hashing together with fast multi-point polynomial evaluation.

\begin{theorem} \label{thm:mainStreamDens}
	There exists a semi-streaming algorithm in dynamic graph streams for the densest subgraph problem with space $\tilde{O}(n)$ which gives a $(1-\epsilon)$-approximate solution, with probability $1-1/n$. The update time is $\tilde{O}(1)$. 
\end{theorem}

In addition, our algorithm can be implemented using an oracle that provides direct access to a uniformly sampled edge. This algorithm makes $\tilde{O}(n)$ queries to the oracle and finds the densest subgraph on a graph with $n$ vertices and $\tilde{O}(n)$ edges. Therefore, by replacing the $m$ in the running time of the algorithm of Charikar \cite{charikar2000greedy} with $\tilde{O}(n)$, we achieve $\tilde{O}(n)$ running time, which is tight up to logarithmic factors. To the best of our knowledge, this is the fastest constant approximation algorithm for the densest subgraph problem in the RAM model, given such an oracle. 

\begin{theorem} \label{thm:mainRAMDens}
	Suppose we have oracle access to the edge set of the input graph with the ability to sample an edge uniformly at random. There exists an algorithm for the densest subgraph problem running in time $\tilde{O}(n)$, which gives a $(0.5-\epsilon)$-approximate solution, with probability $1-1/n$.
\end{theorem}

Next, we extend our results to a general family of graph optimization problems that we call {\it heavy subgraph problems}. Interestingly, we show that by uniformly sampling edges we obtain enough information about the solution of any heavy subgraph problem. Since the solution of a heavy subgraph problem itself may be as large as the whole graph, here we just claim that we can estimate the size of the optimum solution. However, in some cases, like for the densest subgraph problem, it might be possible to also obtain the optimum solution itself, and not just the size, from the sampled graph.

A graph optimization problem is defined by, an input graph $G$, a set of feasible solutions $Sol_G$, which are subgraphs of $G$, and an objective function $f:Sol\rightarrow \mathbb{R}$. In a graph optimization problem we aim to find a solution $sol\in Sol_G$ that maximizes $f$. In fact, the number of feasible solutions for a graph $G$ may be exponential in the size of $G$.
We say a graph optimization problem on graph $G$ is a \emph{$(\gamma,l)$-heavy subgraph problem} if there exist $l$ sets $Sol_G^1,Sol_G^2,\dots,Sol_G^l$, such that $Sol_G = \cup_{k=1}^{l} Sol_G^k$ and for any $k$:

\begin{itemize}
	\item \textbf{Local Linearity:} There are $l$ numbers $f_1\geq f_2\geq \dots\geq f_l = 1$, such that for any $1\leq k\leq l$ and any solution $sol\in Sol^k_G$, we have $f(sol)=f_k\cdot|E_{sol}|$, where $E_{sol}$ is the edge set of $sol$.
	
	\item \textbf{Hereditary Property:} For any spanning subgraph $H\subseteq G$, we have $sol_H\in Sol_H^k$ if and only if there exists a solution $sol_G\in Sol_G^k$ such that $sol_H= sol_G \cap H$. 
	
	\item \textbf{$\gamma$ Bound:} $\gamma$ is chosen such that the optimum solution is lower bounded by $\gamma \log(|Sol_G^k|) f_k \frac m n$.
\end{itemize}

Let $\mathcal{P}(\gamma,l)$ be a heavy subgraph problem, and let $Alg$ be an $\alpha$-approximation algorithm for $\mathcal{P}$. Algorithm \ref{alg:GeneralFamily} samples $O(\frac {n \delta \log(l)} { \gamma \epsilon^2})$ edges of the input graph and runs $Alg$ on the sampled graph. Interestingly, the following Theorem shows Algorithm \ref{alg:GeneralFamily} is an $(\alpha-\epsilon)$-approximation algorithm for $\mathcal{P}$ on $G$.

\begin{theorem}\label{thm:mainGF}
	Let $\mathcal{P}(\gamma,l)$ be a heavy subgraph problem. Let $G$ be an arbitrary graph $G$, and let $Alg$ be an $\alpha$-approximation algorithm for $\mathcal{P}$. With probability $1-e^{-\delta}$, Algorithm \ref{alg:GeneralFamily} is an $(\alpha-\epsilon)$-approximation algorithm for $\mathcal{P}$ on $G$, using $O(\frac {n \delta \log(l)} { \gamma \epsilon^2})$ space.
\end{theorem}

Finally, in section \ref{sec:App}, we show several applications of Theorem \ref{thm:mainGF}. Indeed, we show that directed densest subgraph, densest bipartite subgraph, $d$-max cut and $d$-sum-max clustering all fits in the general family of heavy subgraph problems, and thus, Theorem \ref{thm:mainGF} holds for them.

\begin{theorem}\label{thm:Apps} The following statements hold.
\begin{itemize}
	\item Densest bipartite subgraph is a $(\gamma= \frac{2}{\log(n)+1}, l= n)$ heavy subgraph problem.
	\item Directed densest subgraph is a $(\gamma=\frac{1}{2\sqrt{n}\log(n)} , l= n^2)$ heavy subgraph problem.
	\item $d$-max cut is a $(\gamma= \frac{1}{2\log(d)}, l= 1)$ heavy subgraph problem.
	\item $d$-sum-max clustering is a $(\gamma=\frac{n-2d}{n\log(d)} , l= 1)$ heavy subgraph problem.
\end{itemize}
\end{theorem}

In fact, understanding the structure of the problems that can be solved using sampling, and specifically uniform sampling, is a well-motivated challenge, which was highlighted as a direction in the IITK Workshop on Algorithms for Data Streams in 2006. Our structural results, as well as our characterization for heavy subgraph problems, give partial answers to this open question in the context of graphs. 

In simultaneous and independent work McGregor et al. \cite{MFCSMcGregor} present a single pass $(1-\epsilon)$-approximation algorithm for the densest sugraph problem in the dynamic graph streaming model with update time $\tilde{O}(1)$, that uses $\tilde{O}(n)$ space. 

%% file: DensestSubgraph.tex
In this section we analyze a simple algorithm for the densest subgraph problem. This algorithm simply samples $\tilde{O}(n)$ edges uniformly at random (without replacement) and solves the problem on the sampled subgraph, where factors of $1/\epsilon$, $\delta$ and $\log(n)$ are hidden in the $\tilde{O}$ notation. 
Interestingly, we show that, with probability $1-m^{-\delta}$, this algorithm gives a $(1-\epsilon)$-approximate solution using $\tilde{O}(n)$ bits of space.

In addition, this algorithm can be implemented with a running time of $\tilde{O}(n)$, using oracle access to the edge set, with the ability to sample an edge uniformly at random. This works by combining our sampling guarantee with the algorithm of
Charikar \cite{charikar2000greedy}. 

To the best of our knowledge, our $(0.5-\epsilon)$-approximation algorithm is the fastest constant approximation algorithm for the densest subgraph problem. This algorithm can be implemented in the streaming setting with insertions and deletions to the edges (the \emph{strict turnstile} or \emph{dynamic graph} stream model), as well as the RAM model with oracle access to random edge samples (we discuss the latter model below
in the context of improving the running time of the best offline algorithms).

Before stating our main lemma first we need the following generalized version of the Chernoff inequality, that holds for negatively correlated random variables. We say Boolean random variables $x_1,x_2,\dots,x_r$ are negatively correlated if for any arbitrary subset $S$ of $\{x_1,x_2,\dots,x_r\}$, and any arbitrary $a\in S$ we have $\Pr(a=1| \forall_{b\in S-{a}} b=1) \leq \Pr(a=1)$ \cite{esfandiari2014online}.

\begin{lemma}[\cite{panconesi1997randomized}]\label{lm:chernoff}
Let $x_1,x_2,\dots,x_r$ be a sequence of negatively correlated boolean (i.e., $0$ or $1$) random variables, and let $X=\sum_{i=1}^{r}x_i$. We have 
\begin{align*}
\Pr\left(|X-\E[X]| \geq \epsilon \E[X] \right) \leq 3 \exp(-\epsilon^2\E[X]/3).
\end{align*}
\end{lemma}

To use the above generalized version of the Chernoff bound, we need to show that our random variables are negatively correlated. The following lemma becomes very useful to show the random variables are negatively correlated.

\begin{lemma}\label{lm:negative}
Let $x_1,x_2,\dots,x_r$ be a sequence of Boolean random variables, such that, exactly $t$ of them are chosen to be $1$ uniformly at random. Then the random variables $x_1,x_2,\dots,x_r$ are negatively correlated. 
\end{lemma}
\begin{proof}
Let $S$ be an arbitrary subset of $\{x_1,x_2,\dots, x_r\}$ and let $a$ be an arbitrary element of $S$. On the one hand, the probability that $a=1$ is $\frac t r$. On the other hand, conditioned that for any element $b\in S \setminus \{a\}$, we have $b=1$; the probability that $a=1$ is $\frac {t-(|S|-1)}{r-(|S|-1)}$. Clearly, $\frac t r \geq \frac {t-(|S|-1)}{r-(|S|-1)}$, which gives us
$\Pr(a=1| \forall_{b\in S-{a}} b=1) \leq \Pr(a=1).$
This means the random variables $x_1,x_2,\dots,x_r$ are negatively correlated. 
\end{proof}

\begin{algorithm*}
\textbf{Input:} A graph $G=(E_G,V)$.

\textbf{Output:} A $(1-\epsilon)$-approximation of the densest subgraph of $G$, w.pr. $1-m^{-\delta}$.

\begin{algorithmic}[1]
    \STATE Set $C=\frac {12n (4+\delta) \log(m)} { \epsilon^2}$
    \IF{$|E|\leq C$}
    	\STATE Find a densest subgraph of $G$ using the algorithm of \cite{charikar2000greedy} (or any other algorithm).
    \ELSE
    	\STATE Sample $C$ edges uniformly at random, without replacement from $G$.
    	\STATE Let $H$ be the sampled graph.
    	\STATE Find a densest subgraph of $H$.
    \ENDIF
\end{algorithmic}

\caption{Finding Densest Subgraph}
\label{alg:DensestSubGraph}
\end{algorithm*}

\begin{lemma}\label{lm:kBoundDS}
Let $G=(V,E_G)$ be a graph with vertex set $V$ and edge set $E_G$. Let $H=(V,E_H)$ be the sampled graph in Algorithm \ref{alg:DensestSubGraph} and let $p=\frac C {|E_G|}$.
We have for any $k$, 
\begin{align*}
\Pr\left( opt(G,k) - den_G(opt(H,k)) \geq \epsilon   opt(G)\right) \leq 6 \exp(k\cdot \log(m) -\frac{p k \epsilon^2  opt(G)} {12} ),
\end{align*}
where $den_G(opt(H,k))$ is the density of the subgraph of $G$ induced by the vertices of $opt(H,k)$.
\end{lemma}

\begin{proof}
	Let $x_e$ be the random variable that indicates whether $e$ exists in $E_H$ or not, and
	let $U$ be an arbitrary subset of $V$ of size $k$. By definition, the number of edges in $H[U]$ is $\sum_{e \in H[U]} 1 = \sum_{e \in G[U]} x_e$. Let us denote this summation by $X_U$. Then, we have
	\begin{align}\label{eq:expectedDS}
		\E[den(H[U])] &= \E \left[ \frac {\sum_{e \in G[U]} x_e}{|U|} \right ] 
		= \frac {\E[\sum_{e \in G[U]} x_e]}{|U|}  
		=  \frac {\sum_{e \in G[U]} p}{|U|}  
		= p \frac {\sum_{e \in G[U]} 1}{|U|} 
		= p\cdot den(G[U])
	\end{align}
	where the first and the last equalities are by definition of density, and the third equality is by definition of $p$.
	Lemma \ref{lm:negative} says that random variables $x_1,x_2,\dots,x_{|E_G|}$ are negatively correlated. Thus, to bound $X_U$, we can apply the following form of the Chernoff bound from Lemma \ref{lm:chernoff}.
	\begin{align*}
		\Pr\left(|X-\E[X]| \geq \epsilon' \E[X] \right) 
		\leq 3 \exp(-\epsilon'^2\E[X]/3)
	\end{align*}
	By setting $\epsilon'=p k \epsilon \frac{ opt(G)} {2 \E[X]}$ we have
	\begin{align*}
		\Pr\left(|X-\E[X]| \geq \frac {p k \epsilon  opt(G)}{2} \right) 
		&\leq 3 \exp(-\frac{p^2 k^2 \epsilon^2  opt(G)^2} {4\E[X]^2} \cdot \frac{\E[X]}{3}) \\
		&\leq 3 \exp(-\frac{p^2 k^2 \epsilon^2  opt(G)^2} {12\E[X]} ) \\
		&\leq 3 \exp(-\frac{p k \epsilon^2  opt(G)} {12} ) &\text{Using Equality \ref{eq:expectedDS}}\\
	\end{align*}
	On the other hand we have
	\begin{align*}
		\Pr\left(|X-\E[X]| \geq \frac{p k \epsilon  opt(G)}{2} \right) &=
		\Pr\left(\frac 1 p \frac {|X-\E[X]|}{k} \geq   \frac{\epsilon}{2}  opt(G) \right)\\ 
		&= \Pr\left(|\frac 1 p \frac {X}{k} -  \frac {\frac 1 p \E[X]}{k}| \geq   \frac{\epsilon}{2}  opt(G) \right) \\
		&= \Pr\left(|\frac 1 p den(H[U]) -  den(G[U])| \geq   \frac{\epsilon}{2}  opt(G) \right) &\text{Using Equality \ref{eq:expectedDS}}\\
	\end{align*}
	Therefore, we have
	\begin{align} \label{eq:Ubound}
		\Pr\left(|\frac 1 p den(H[U]) -  den(G[U])| \geq   \frac{\epsilon}{2}  opt(G) \right) \leq 3 \exp(-\frac{p k \epsilon^2  opt(G)} {12} ).
	\end{align}
	If we set $U$ to be the vertex set of $Opt(G,k)$, Inequality \ref{eq:Ubound} says that 
	\begin{align*}
		\Pr\left(  opt(G,k) - \frac 1 p den_H(opt(G,k))  \geq \frac{\epsilon}{2}   opt(G)\right) \leq 3 \exp(-\frac{p k \epsilon^2  opt(G)} {12} )
	\end{align*}
	which immediately gives us
	
	\begin{align}\label{eq:kbound1}
		\Pr\left(  opt(G,k) - \frac 1 p opt(H,k)  \geq \frac{\epsilon}{2}   opt(G)\right) \leq 3 \exp(-\frac{p k \epsilon^2  opt(G)} {12} ).
	\end{align}
	On the other hand, Inequality \ref{eq:Ubound} says that for each selection of $U$, with probability $1-3 \exp(-\frac{p k \epsilon^2  opt(G)} {12})$, we can upper bound $\frac 1 p den(H[U])$  by 
	$den(G[U]) + \frac{\epsilon}{2} opt(G)$. 
	As we have ${m \choose k}$ such choices, applying a union bound we have
	\begin{align*}
		\Pr\left(\forall_{U:|U|=k}  \frac 1 p den(H[U]) - den(G[U])  \geq \frac{\epsilon}{2}   opt(G)\right) \leq 3{m \choose k} \exp(-\frac{p k \epsilon^2  opt(G)} {12} ).
	\end{align*}
	If we set $U$ to $opt(H,k)$ we have
	\begin{align}\label{eq:kbound2}
		\Pr\left(\frac 1 p opt(H,k) - den_G(opt(H,k))  \geq \frac{\epsilon}{2}   opt(G)\right) \leq 3{m \choose k} \exp(-\frac{p k \epsilon^2  opt(G)} {12} ).
	\end{align}
	Therefore, by combining Inequalities \ref{eq:kbound1} and \ref{eq:kbound2} and applying the union bound we have
	\begin{align*}
		\Pr\left(   opt(G,k)- den_G(opt(H,k)) \geq \epsilon   opt(G)\right) &\leq 3({m \choose k}+1) \exp(-\frac{p k \epsilon^2  opt(G)} {12} )\\
		& \leq 3({m^k}+1) \exp(-\frac{p k \epsilon^2  opt(G)} {12} )\\
		& \leq 6 \exp(k\cdot \log(m) -\frac{p k \epsilon^2  opt(G)} {12} ).
	\end{align*}
\end{proof}

In fact, the density of the densest subgraph of a graph $G$ is at least as much as the density of $G$ itself. Hence, one can lower bound $opt(G)$ by the density of $G$ which is $\frac m n$. The following states this fact.

\begin{fact} \label{fc:boundOpt}
The density of the densest subgraph of $G$ is at least $\frac m n$, where $n$ is the number of vertices in $G$ and $m$ is the number of edges.
\end{fact}

The following theorem bounds the approximation ratio of Algorithm \ref{alg:DensestSubGraph}. 

\begin{theorem}\label{thm:mainDS}
With probability $1-m^{-\delta}$, Algorithm \ref{alg:DensestSubGraph} is a $(1-\epsilon)$-approximation algorithm for the densest subgraph problem.
\end{theorem}

\begin{proof}
	Recall that Lemma \ref{lm:kBoundDS} states for an arbitrary fixed $k$, with probability at least $1-6 \exp(k\cdot \log(m) -\frac{p k \epsilon^2  opt(G)} {12} )$ we have 
	$opt(G,k) - den_G(opt(H,k)) \leq \epsilon   opt(G)$. Using a union bound this holds for all $1\leq k\leq n$, with probability $1- \sum_{k=1}^n 6 \exp(k\cdot \log(m) -\frac{p k \epsilon^2  opt(G)} {12} )$. Thus, for some value of $k$, with probability $1-n \cdot 6 \exp(k\cdot \log(m) -\frac{p k \epsilon^2  opt(G)} {12} )$ we have $opt(G) - den_G(opt(H)) \leq \epsilon   opt(G)$, which means that Algorithm \ref{alg:DensestSubGraph} returns a $(1-\epsilon)$-approximation. Now, if we set $p$ to $\frac {12n (4+\delta) \log(m)} { \epsilon^2  m}$, or equivalently set $C$ in Algorithm \ref{alg:DensestSubGraph} to $\frac {12n (4+\delta) \log(m)} { \epsilon^2}$, we have
	\begin{align*}
		&1-n\times 6 \exp(k\cdot \log(m) -\frac{p k \epsilon^2  opt(G)} {12} ) \\
		&= 1-n\times 6 \exp(k\cdot \log(m) -\frac {12n (4+\delta) \log(m)} { \epsilon^2  m} \times\frac{ k \epsilon^2  opt(G)} {12} ) \\
		& \geq 1-n\times 6 \exp(k\cdot \log(m) -\frac {12n (4+\delta) \log(m)} { \epsilon^2  m} \times \frac{ k \epsilon^2  m} {12n} ) & \text{From Fact \ref{fc:boundOpt}}\\
		& = 1-n\times 6 \exp(k\cdot \log(m) -(4+\delta)k\cdot \log(m))\\
		& = 1-n\times 6 \exp(-(3 + \delta) k\cdot \log(m))\\
		& \geq 1- \exp(\log(n)+2 -(3 + \delta) k\cdot \log(m)) &\text{Since $e^2\geq 6$}\\
		& \geq 1- \exp(-\delta\cdot k\cdot \log(m))\\
		& \geq 1- \exp(-\delta\cdot \log(m))\\
		&=1- m^{-\delta}.
	\end{align*}
\end{proof}

One can use $L_0$-samplers \cite{jst11} 
to sample the $C$ edges in Algorithm \ref{alg:DensestSubGraph} in a dynamic graph stream. However, maintaining these $L_0$ samplers may need an update time as large as $C$.
In Lemma \ref{lm:updatetime} we show how to sample $C$ edges with $\tilde{O}(1)$ update time using the notion of min-wise independent hash functions.
\begin{definition}
Given $\epsilon > 0$, 
we say a hash function $h:[1,n]\rightarrow [1,n]$ is {\bf $\epsilon$-approximately $t$-min-wise independent} on a subset $X$ of $\{1, 2, \ldots, n\}$ 
if for any $Y\subseteq X$ such that $|Y|=t$ we have
\begin{align*}
Pr(\max_{y \in Y} h(y) < \min_{z\in X-Y} h(z)) = \frac1 {{|X| \choose |Y|}} (1 \pm \epsilon).
\end{align*}
\end{definition}

\begin{theorem}\label{lm:min-wise}(Theorem 1.1 of \cite{feigenblat2011exponential})
There exist constants $c', c'' > 1$ such that for any $X \subseteq \{1, 2, \ldots, n\}$ of size at most
$\epsilon n / c'$, any $c'' (t \log \log (1/\epsilon) + \log(1/\epsilon))$-wise independent
family of functions $h$ is $\epsilon$-approximately $t$-min-wise independent on $X$. 

In addition, one can evaluate $h$ on $t$ items simultaneously
in total time $t \cdot (\log(t/\epsilon))^{O(1)}$ by using fast multipoint polynomial
evaluation (see, e.g., Theorem 13 of \cite{knpw11}, where this idea was used for a different
streaming problem). Further, one can spread the evaluation
of $h$ on $t$ items evenly across the next $t$ stream updates, converting this amortized
$(\log(t/\epsilon))^{O(1)}$ update time to a worst-case update time. 
\end{theorem}
%

We need the following generalized version of the Chernoff bound.

\begin{lemma}[\cite{schmidt1995chernoff}]\label{lm:Kchernoff}
Let $X$ be the sum of $t$-wise independent Boolean random variables. For any $\epsilon\geq 1$ such that $t \geq \lfloor \epsilon^2\E[X]e^{-1/3} \rfloor$, we have
\begin{align*}
Pr(|X-\E[X]|>\epsilon \E[X]) < \exp(-\lfloor \epsilon^2 \E[X]/3 \rfloor).
\end{align*}
\end{lemma}

\begin{lemma}\label{lm:updatetime}
For any number $C \geq n$ of edge samples and constant $1\leq \delta$, we can sample $C$ edges in a dynamic stream (a stream with insertions and deletions to the edges), such that the statistical distance of our sampled edges to a uniformly random (without replacement) sample of $C$ edges is $O(e^{-\delta})$. This sampling algorithm uses $\tilde{O}(C)$ space and has update time $\tilde{O}(1)$.
\end{lemma}
\begin{proof}
We apply Theorem \ref{lm:min-wise} with the $n$ of that theorem equal to ${n \choose 2} c' / \epsilon$, where $\epsilon$ is set to $e^{-\delta}$. We label each possible edge of our graph with a number in $\{1, 2, \ldots, {n \choose 2}\}$ and extend the domain and range of $h$ to $\{1, 2, \ldots, n^2 c' / \epsilon\}$. Then, we apply the $X$ of Theorem \ref{lm:min-wise} to the specific subset of $m \leq {n \choose 2}$ edges in our input graph. We apply Theorem \ref{lm:min-wise}
with $t = \Theta(C \delta)$. Theorem \ref{lm:min-wise} and the definition of min-wise independence imply that the statistical distance of the subset $Y$ of $t$ minimum hash values of our edges under $h$ is within $e^{-\delta}$ from $t$ uniformly random samples without replacement. Indeed, the probability of choosing any set $Y$ is $\frac{1}{{|X| \choose |Y|}} (1 \pm \epsilon)$ rather than $\frac{1}{{|X| \choose |Y|}}$ had we had full independence, so summing the absolute values of these differences gives $\ell_1$-distance at most $\epsilon = e^{-\delta}$, and so statistical distance at most $e^{-\delta}/2$ between the distributions. 
Note that Theorem \ref{lm:min-wise} implies
that $h$ is also $\Theta(C \delta)$-wise independent (in the standard sense, not the min-wise
sense), and so $h$ is $\Omega(n)$-wise independent using our assumptions that $C \geq n$ and
$\delta \geq 1$. 

Now we show how 
to maintain the $C$ edges with the smallest hash values under $h$ in a dynamic stream. 
We use $\log(n^2)$ sparse recovery data structures, $sp_1,sp_2,\dots sp_{2\log(n)}$ each to recover $9\delta C$ edges. Note that $n^2$ is an upper bound on the total number of distinct edges in our graph. Upon the update (insertion or deletion) of an edge $e$, we update each sparse recovery structure $s_i$ such that $i \in [1, \lfloor \log_2(h(e))\rfloor]$. For the sparse recovery structure, we use the data structure of \cite{glps10}, which has $\tilde{O}(1)$ update time, space $\tilde{O}(C)$ and succeeds with probability $1-1/n^2$ in returning all of the non-zero items
in a vector for which it is applied to, provided this number of non-zero items is at most $9\delta C$. 

Let $Z_i$ be the number of distinct edges which hash under $h$ to the $i$-th sparse recovery
data structure $sp_i$. Then if $m \geq 4\delta C$, there always exists one of the $sp_i$ such that $4\delta C \leq \E[Z_i]\leq 8\delta C$ (if $m < 4\delta C$, we can store the entire graph in $\tilde{O}(C)$ bits of space). Moreover, $Z_i$ for any given $i$ is fairly concentrated around its expectation, since Theorem \ref{lm:min-wise} implies that $h$ is also $\Theta(C \delta)$-wise independent, so we can bound its deviation using the generalized version of the Chernoff bound given by Lemma \ref{lm:Kchernoff}. Consider any $Z_i$ for which $4\delta C \leq \E[Z_i] \leq 8\delta C$. Then we have 
\begin{align*}
\Pr(Z_i< 3\delta C \text{ or }Z_i > 9\delta C) &\leq \Pr(|Z_i-\E[Z_i]|>\delta C)  
\leq \Pr(|Z_i-\E[Z_i]|> \frac 1 8 \E[x])\\
& \leq \exp(-\left \lfloor (\frac 1 8)^2 \frac {\E[Z_i]} 3 \right \rfloor) \leq  \exp(-\Theta(C)) 
\end{align*}
Using that $C \geq n$, the error probability is $\exp(-n)$. We also run an $\ell_0$-estimation algorithm, to estimate
each $Z_i$ up to a multiplicative factor of $1.1$, with total space $\tilde{O}(1)$ and update time $\tilde{O}(1)$ \cite{cdim03}, 
and with failure probability $1/n^{\Omega(1)}$. It follows
from the above that for an $i$ for which $4\delta C \leq \E[Z_i] \leq 8\delta C$, we have that $Z_i \in [3\delta C, 9 \delta C]$ with
probability $1-\exp(-n)$. It follows that our $\ell_0$-estimate for this value of $Z_i$ will be in $[C, 10\delta C]$ with
probability $1-1/n^{\Omega(1)}$. Hence, from $sp_i$, with probability $1-1/n^{\Omega(1)}$ we will recover all values that hash to $sp_i$,
and as argued above these are within statistical distance $e^{-\delta}$ from uniform.
\end{proof}

Combining Lemma \ref{lm:updatetime} with Theorem \ref{thm:mainDS} proves Theorem \ref{thm:mainStreamDens}.

To the best of our knowledge, the fastest known $0.5$-approximation algorithm for the densest subgraph problem has a running time of $O(m+n)$. However, here we need to find the densest subgraph of a sampled graph with at most $C=\tilde{O}(n)$ edges. Thus, the running time of our algorithm is $\tilde{O}(n)$. To the best of our knowledge, this is the fastest $(0.5-\epsilon)$-approximation algorithm for the densest subgraph problem.
Theorem \ref{thm:mainRAMDens} states this fact.

%% file: GeneralFamily.tex
Here, we extend our results to the heavy subgraph problems. 
Specifically, we show that given an offline $\alpha$-approximation algorithm for a $(\gamma,l)$-heavy subgraph problem, $\mathcal{P}(\gamma,l)$,
Algorithm \ref{alg:GeneralFamily} is an $(\alpha-\epsilon)$-approximation for $\mathcal{P}(\gamma,l)$.
Later, in Section \ref{sec:App}, we show several applications of this algorithm.
In this section, we denote the solution in $Sol_G^k$ that maximizes $f$ by $opt(G,k)$.

\begin{algorithm*}
\textbf{Input:} A graph $G$, a heavy subgraph problem $\mathcal{P}(\gamma, l)$ and an $\alpha$ approximation algorithm $Alg$ for $\mathcal{P}$.

\textbf{Output:} An $\alpha - \epsilon$ estimator of $\mathcal{P}$ on graph $G$, w.pr. $1-m^{-\delta}$.

\begin{algorithmic}[1]
    \STATE Set $C=\frac {12n (4+\delta) \log(l)} { \gamma \epsilon^2}$
    \IF{$|E|\leq C$}
    	\STATE Return $Alg(G)$.
    \ELSE
    	\STATE Sample $C$ edges uniformly at random, without replacement from $G$.
    	\STATE Let $H$ be the sampled graph.
    	\STATE Return $\frac 1 p Alg(H)$.
    \ENDIF
\end{algorithmic}

\caption{A General Algorithm}
\label{alg:GeneralFamily}
\end{algorithm*}

The following lemma is the generalized version of Lemma \ref{lm:kBoundDS}. 

\begin{lemma}\label{lm:kBoundGF}
Let $G$ be the input graph and let $\mathcal{P}(\gamma,l)$ be a heavy subgraph problem and let $Alg$ be an $\alpha$-approximation algorithm for problem $\mathcal{P}$. Let $H=(V,E_H)$ be the sampled graph by Algorithm \ref{alg:GeneralFamily} and let $p=\frac C {m}$.
We have 
\begin{align*}
\Pr\left(  \alpha opt(G,k)-  f_G(Alg(H))) \geq  \epsilon   opt(G)\right) \leq 6 \exp(\log(|Sol_G^k|)-\frac{p \epsilon^2  opt(G)} {12f_k} ),
\end{align*}
where $f_G(Alg(H))$ is the objective value of $G$ on a solution $sol_G$ such that $Alg(H)=sol_G\cap H$.
\end{lemma}

\begin{proof}
	We prove this lemma in a similar way to that of Lemma \ref{lm:kBoundDS}. Again here we define $x_e$ to be the random variable that indicates whether $e$ exists in $E_H$ or not. 
	However, here we let $sol_G$ be an arbitrary solution from $Sol_G^k$. 
	Let $sol_H$ be a spanning subgraph of $sol_G$ that contains the edges that appear in both $sol_G$ and $H$. The hereditary property says that $sol_H\in Sol_H^k$.
	
	By definition, the number of edges in $sol_H$ is $\sum_{e \in sol_H} 1 = \sum_{e \in sol_G} x_e$. We denote this summation by $X$. Using the local linearity property we have
	\begin{align}\label{eq:expectedGF}
		\E[f(sol_H)] = \E[ f_k {\sum_{e \in sol_G} x_e} ] 
		=  f_k {\sum_{e \in sol_G} \E[x_e]}   
		=  f_k {\sum_{e \in Sol_G} p}  
		= p\cdot f_k {\sum_{e \in Sol_G} 1} 
		= p\cdot f(sol_G).
	\end{align}
	where the first and the last equalities are by Local Linearity, and the third equality is by definition of $p$.
	Again we can use Lemma \ref{lm:negative} to claim that the random variables $x_1,x_2,\dots,x_{m}$ are negatively correlated. Thus, to bound $X_U$, we can apply the following form of the Chernoff bound from Lemma \ref{lm:chernoff}.
	\begin{align*}
		\Pr\left(|X-\E[X]| \geq \epsilon' \E[X] \right) 
		\leq 3 \exp(-\epsilon'^2\E[X]/3)
	\end{align*}
	By setting $\epsilon'=p  \epsilon \frac{ opt(G)} {2f_k \E[X]}$ we have
	\begin{align*}
		\Pr\left(|X-\E[X]| \geq \frac {p \epsilon  opt(G)}{2f_k} \right) 
		&\leq 3 \exp(-\frac{p^2 \epsilon^2  opt(G)^2} {4f_k^2\E[X]^2} \cdot \frac{\E[X]}{3}) \\
		&\leq 3 \exp(-\frac{p^2 \epsilon^2  opt(G)^2} {12f_k^2\E[X]} ) \\
		&\leq 3 \exp(-\frac{p \epsilon^2  opt(G)} {12f_k} ) &\text{Using Equality \ref{eq:expectedGF}} \\
	\end{align*}
	On the other hand we have
	\begin{align*}
		\Pr\left(|X-\E[X]| \geq \frac{p \epsilon  opt(G)}{2f_k} \right) &=
		\Pr\left(\frac 1 p f_k {|X-\E[X]|} \geq   \frac{\epsilon}{2}  opt(G) \right)\\ 
		&= \Pr\left(|\frac 1 p f_k {X} -   {\frac 1 p f_k \E[X]}| \geq   \frac{\epsilon}{2}  opt(G) \right) \\
		&= \Pr\left(|\frac 1 p f(sol_H) -  f(sol_G)| \geq   \frac{\epsilon}{2}  opt(G) \right) &\text{Using Equality \ref{eq:expectedGF}} \\
	\end{align*}
	Therefore, we have
	\begin{align} \label{eq:UboundGF}
		\Pr\left(|\frac 1 p f(sol_H) -  f(sol_G)| \geq   \frac{\epsilon}{2}  opt(G) \right) \leq 3 \exp(-\frac{p \epsilon^2  opt(G)} {12f_k} )
	\end{align}
	If we set $sol_G$ to be the vertex set of $Opt(G,k)$, Inequality \ref{eq:UboundGF} states that 
	\begin{align*}
		\Pr\left(  opt(G,k) - \frac 1 p f_H(opt(G,k))  \geq \frac{\epsilon}{2}   opt(G)\right) \leq 3 \exp(-\frac{p \epsilon^2  opt(G)} {12f_k} )
	\end{align*}
	which immediately gives us 
	\begin{align}\label{eq:kbound1GF}
		\Pr\left(  opt(G,k) - \frac 1 p opt(H,k)  \geq \frac{\epsilon}{2}   opt(G)\right) \leq 3 \exp(-\frac{p \epsilon^2  opt(G)} {12f_k} ).
	\end{align}
	On the other hand, Inequality \ref{eq:UboundGF} states that for each selection of $sol_G$, with probability $1-3 \exp(-\frac{p \epsilon^2  opt(G)} {12f_k} )$, we can upper bound $\frac 1 p f(sol_H)$  by 
	$f(sol_G) + \frac{\epsilon}{2} opt(G)$. 
	Indeed, we have $|Sol_G^k|$ such choices. Thus, by applying a union bound we have
	\begin{align*}
		\Pr\left(\forall_{sol_G\in Sol_G^k} \frac 1 p f(sol_H) -  f(sol_G) \geq   \frac{\epsilon}{2}  opt(G) \right) \leq 3 |Sol_G^k|\exp(-\frac{p \epsilon^2  opt(G)} {12f_k} ).
	\end{align*}
	If we select $sol_G$, using the hereditary property, such that $sol_H=Alg(H)$, we have
	\begin{align*}
		\Pr\left(\frac 1 p Alg(H) - f_G(Alg(H))  \geq \frac{\epsilon}{2}   opt(G)\right) \leq  3 |Sol_G^k|\exp(-\frac{p \epsilon^2  opt(G)} {12f_k} ).
	\end{align*}
	Now, given that $Alg(H)\geq \alpha Opt(H)$, we have
	\begin{align}\label{eq:kbound2GF}
		\Pr\left(\frac 1 p opt(H) - \frac 1 {\alpha} f_G(Alg(H))  \geq \frac 1 {\alpha} \frac{\epsilon}{2}   opt(G)\right) \leq  3 |Sol_G^k|\exp(-\frac{p \epsilon^2  opt(G)} {12f_k} ).
	\end{align}
	
	Therefore, by combining Inequalities \ref{eq:kbound1GF} and \ref{eq:kbound2GF} and applying the union bound we have
	\begin{align*}
		\Pr\left(   opt(G,k)- \frac 1 {\alpha} f_G(Alg(H))) \geq \frac 1 {\alpha} \epsilon   opt(G)\right) &\leq 3 (|Sol_G^k|+1)\exp(-\frac{p \epsilon^2  opt(G)} {12f_k} )\\
		&\leq 6 \exp(\log(|Sol_G^k|)-\frac{p \epsilon^2  opt(G)} {12f_k} )
	\end{align*}
\end{proof}

Now, we are ready to prove Theorem \ref{thm:mainGF}.

\begin{proofof}{Theorem \ref{thm:mainGF}}
Lemma \ref{lm:kBoundGF} together with Equality \ref{eq:expectedGF} imply that for each $k$, with probability at least $1-6 \exp(\log(|Sol_G^k|)-\frac{p \epsilon^2  opt(G)} {12f_k} )$ we have 
$ \alpha opt(G,k)-  \frac 1 p Alg(H)) \geq  \epsilon   opt(G)$. 
By a union bound, this holds for all $1\leq k\leq l$, with probability $1- \sum_{k=1}^{l} 6 \exp(\log(|Sol_G^k|)-\frac{p \epsilon^2  opt(G)} {12f_k} )$.
Thus, for some $k$ with probability $1-6l\cdot \exp(\log(|Sol_G^k|)-\frac{p \epsilon^2  opt(G)} {12f_k} )$ we have $ \alpha opt(G)-  \frac 1 p Alg(H)) \geq  \epsilon   opt(G)$, which means that Algorithm \ref{alg:GeneralFamily} outputs a $(1-\epsilon)$-approximation. 

Now, if we set $p$ to $\frac {12n (4+\delta) \log(l)} {\gamma \epsilon^2  m}$, or equivalently set $C$ in Algorithm \ref{alg:GeneralFamily} to $\frac {12n (4+\delta) \log(l)} { \gamma \epsilon^2}$, we have
\begin{align*}
&1-6l \cdot \exp(\log(|Sol_G^k|)-\frac{p \epsilon^2  opt(G)} {12f_k} )\\
&=1-6l \cdot \exp(\log(|Sol_G^k|)-\frac {12n (4+\delta) \log(l)} {\gamma \epsilon^2  m} \frac{ \epsilon^2  opt(G)} {12f_k} )\\
&\leq 1-6l \cdot \exp(\log(|Sol_G^k|)-\frac {12n (4+\delta) \log(l)} {\gamma \epsilon^2  m} \frac{ \epsilon^2 } {12f_k} \gamma \log(|Sol_G^k|) f_k \frac m n) & \text{From $\gamma$ Bound}\\
&=1-6l \cdot \exp(\log(|Sol_G^k|)- { (4+\delta) \log(l)}  \log(|Sol_G^k|)  ) \\
&<1-  \exp(\log(l)+2+\log(|Sol_G^k|)- { (4+\delta) \log(l)}  \log(|Sol_G^k|)  ) \\
&\leq 1- \exp(- { \delta \log(l)}  \log(|Sol_G^k|)  ) \\
& = 1- e^{-\delta}
\end{align*}
\end{proofof}

%% file: Applications.tex
There are several problems that fit into the class of heavy subgraph problems. Some examples are densest bipartite subgraph, directed densest subgraph, $d$-max cut, and $d$-sum-max clustering. In this section we define each of these problems and prove that each satisfies the properties required of a heavy subgraph problem. 

\subsection{Densest Bipartite Subgraph}\label{subsec:DBS}
In the \emph{densest bipartite subgraph} problem, we are given a graph general $G$ and we aim to find a bipartite subgraph $sol$ of $G$ with the maximum density. Let $opt$ be a densest bipartite subgraph of $G$, with parts $A_{opt}$ and $B_{opt}$. Then $opt$ contains all edges of $G$ that are between $A_{opt}$ and $B_{opt}$. We call such a subgraph feasibly maximal and without loss of generality, restrict all of the solutions to be feasibly maximal. In fact, we can indicate a feasibly maximal solution $sol$ by its two parts $A_{sol}$ and $B_{sol}$.

\begin{proofof}{First part of Theorem \ref{thm:Apps}}
For normalization purposes, we increase the value of the objective function by a factor of $n$. Without loss of generality define the density to be $n\cdot \frac{|E_{sol}|}{|V_{sol}|}$.
We set $l=n$, and for any $1\leq k\leq l$, we let $Sol_G^k$ be the set of all solutions $sol$ such that $|A_{sol}|+|B_{sol}|=n-k+1$. Thus, we have $|Sol_G^k|={n \choose {n-k+1}}2^{n-k+1}$. 
\\
\textbf{Local Linearity:}
The density of a solution $sol$ in $Sol_G^k$ is $n\frac{|E_{sol}|}{n-k+1}= |E_{sol}|\frac{n}{n-k+1}$. Thus, we can set $f_k=\frac{n}{n-k+1}$, which is increasing in $k$ and we have $f_1=\frac{n}{n-1+1}=1$, as desired.
\\
\textbf{Hereditary Property:} Let $H$ be a spanning subgraph of $G$. For any solution $sol\in Sol_G^k$, the intersection of $sol$ and $H$ remains bipartite and feasibly maximal, and thus, is a solution in $Sol_H$. Moreover, by definition, the number of vertices of this intersection is the same as $sol$. Thus, it belongs to $Sol_H^k$. On the other hand, for any solution $sol \in Sol_H^k$ with parts $A_{sol}$ and $B_{sol}$, let $sol'$ be the bipartite maximal subgraph of $G$ between the partitions $A_{sol}$ and $B_{sol}$. By definition, $sol'\in Sol_G^k$ and clearly $sol'$ satisfies $sol=sol'\cap H$.
\\
\textbf{$\gamma$ Bound:} In fact, $G$ contains a bipartite subgraph that contains at least $\frac{m}{2}$ edges. The density of this subgraph is at least $n\frac{m/2}{n}=\frac m 2$. On the other hand we have, $f_k=\frac{n}{n-k+1}$ and $|Sol_G^k|={n \choose {n-k+1}}2^{n-k+1}$. Thus, if we set $\gamma$ to $\frac{2}{\log(n)+1}$, we have
\begin{align*}
\gamma \log(|Sol_G^k|) f_k \frac m n & = \frac{2}{\log(n)+1} \log({n \choose {n-k+1}}2^{n-k+1}) \frac{n}{n-k+1} \frac m n \\
&\leq \frac{2}{\log(n)+1} ((n-k+1)\log(n)+(n-k+1)) \frac{n}{n-k+1} \frac m n \\
&\leq \frac{2}{\log(n)+1} (\log(n)+1)   m \\
&= \frac m 2
 \leq opt.
\end{align*}
\end{proofof}
\subsection{Directed Densest Subgraph}
In the \emph{directed densest subgraph} problem we are given a directed graph $G$, we want to find two not necessarily disjoint sets $A,B\subseteq V_G$, to maximize $\frac{|E(A,B)|}{\sqrt{|A|\cdot|B|}}$, where, $E(A,B)$ is the set of all edges $(u,v)\in E_G$, such that $u\in A$ and $v\in B$.

\begin{proofof}{Second part of Theorem \ref{thm:Apps}}
For normalization, we increase the objective function by a factor of $n$ and define it as $n\frac{E(A,B)}{\sqrt{|A|\cdot|B|}}$.
For simplicity, here we index the solution sets using a pair of indices $i$ and $j$.
Here, $Sol_G^{i,j}$ contains any solution $sol=(A,B)$ such that $i=|A|$ and $j=|B|$. Thus, we have $|Sol_G^{i,j}|= {n \choose i} {n \choose j}$, and $l=n^2$.
\\
\textbf{Local Linearity:} In fact, for a solution $sol=(A,B)\in Sol_G^{i,j}$, the directed density of $sol$ is $n\frac{|E(A,B)|}{\sqrt{|A|\cdot|B|}}=n\frac{|E(A,B)|}{\sqrt{i\cdot j}} =|E(A,B)| \frac{n}{\sqrt{i\cdot j}}$. Thus, we can define, $f_{i,j}= \frac{n}{\sqrt{i\cdot j}}$, and we have $min_{i,j}(f_{i,j})=\frac{n}{\sqrt{n\cdot n}}=1$, as desired.
\\
\textbf{Hereditary Property:} For any spanning subgraph $H\subseteq G$, and any solution $sol=(A,B)\in Sol_G^{|A|,|B|}$, the same sets $A$ and $B$ indicate the intersection of $H$ and $sol$, and thus is a solution in $Sol_H^{|A|,|B|}$.
 On the other hand, for any solution $sol \in Sol_H^{|A|,|B|}$ with sets $A$ and $B$, let $sol'$ be the solution on $G$ corresponds to the sets $A$ and $B$. By definition, $sol'\in Sol_G^{|A|,|B|}$ and clearly $sol'$ satisfies $sol=sol'\cap H$.
\\
\textbf{$\gamma$ Bound:} The directed density of the solution $sol=(V_G, V_G)$ is $n\frac{m}{\sqrt{n\cdot n}}=m$. Therefore, the optimum is lower bounded by $m$.
On the other hand, for any $i$ and $j$ we have $|Sol_G^{i,j}|= {n \choose i} {n \choose j}$ and $f_{i,j}= \frac{n}{\sqrt{i\cdot j}}$. If we set $\gamma$ to $\frac{1}{2\sqrt{n}\log(n)}$ we have
\begin{align*}
\gamma \log(|Sol_G^{i,j}|) f_{i,j} \frac m n & = \gamma \log({n \choose i} {n \choose j}) \frac{n}{\sqrt{i\cdot j}} \frac m n\\
& \leq \gamma (i\cdot \log(n) + j\cdot \log(n)) \frac{n}{\sqrt{i\cdot j}} \frac m n\\
& = \gamma \frac{i\cdot \log(n) + j\cdot \log(n)}{\sqrt{i\cdot j}}  m \\
& = \gamma (\frac{\sqrt{i}\cdot \log(n)}{\sqrt{ j}}+ \frac{\sqrt{j}\cdot \log(n)}{\sqrt{i}})  m \\
& \leq \gamma 2\sqrt{n}\log(n)  m \\
& = \frac{1}{2\sqrt{n}\log(n)} 2\sqrt{n}\log(n)  m \\
& =  m 
\leq opt.
\end{align*}
\end{proofof}
\subsection{$d$-Max Cut}
In the \emph{$d$-max cut} problem, we are given a graph $G$ and are supposed to mark the vertices using $d$ labels to maximize the number of edges with different labels. 

\begin{proofof}{Third part of Theorem \ref{thm:Apps}}
Here we simply let all the solutions be in $Sol_G^1$. Indeed, we have $l=1$ and $|Sol_G^1|= d^n$.
\\
\textbf{Local Linearity:} Clearly we have $f_1=f_l=1$.
\\
\textbf{Hereditary Property:} For any spanning subgraph $H\subseteq G$, and any solution $sol\in Sol_G^1$, the same labeling of $sol$ on $H$ gives us the intersection of $sol$ and $H$. Thus, the intersection of $sol$ and $H$ is a solution in $Sol_H=Sol_H^1$. On the other hand, similarly, for any solution $sol \in Sol_H^1$, the same labeling gives us a solution $sol'\in Sol_G^1$ such that $sol=sol'\cap H$.
\\
\textbf{$\gamma$ Bound:} Again here, if we just use $2$ labels we have a solution with $\frac m 2$ value. Thus, we have $opt\geq \frac m 2$. If we set $\gamma$ to $\frac{1}{2\log(d)}$ we have
\begin{align*}
\gamma \log(|Sol_G^k|) f_k \frac m n & = \frac{1}{2\log(d)} \log(d^n)  \frac m n 
= \frac{1}{2\log(d)} \log(d)  m 
= \frac{m}{2}
\leq opt.
\end{align*}
\end{proofof}

\subsection{$d$-Sum-Max Clustering} \label{subsec:DSMC}
This problem is fairly similar to $d$-max cut. Again we are given a graph $G$ and are supposed to mark the vertices using $d$ labels. However, here we have to use all $d$ colors and want to maximize the number of edges with the same labels. 

\begin{proofof}{Fourth part of Theorem \ref{thm:Apps}}
Again, here we simply let all the solutions be in $Sol_G^1$. So we have $l=1$ and $|Sol_G^1|\leq d^n$.
\\
\textbf{Local Linearity:} We have $f_1=1$.
\\
\textbf{Hereditary Property:}  Consider a spanning subgraph $H\subseteq G$, and let $sol\in Sol_G^1$ be an arbitrary solution. Again, we can use the same labeling as $sol$ on $H$ to get the intersection of $sol$ and $H$. Thus, the intersection of $sol$ and $H$ is a solution in $Sol_H^1$. On the other hand, again, for any solution $sol \in Sol_H^1$, the same labeling gives us a solution $sol'\in Sol_G^1$ such that $sol=sol'\cap H$.
\\
\textbf{$\gamma$ Bound:} 
Suppose we choose $d-1$ vertices uniformly at random and label them with labels $1,2,\dots,d-1$ and label all the other vertices with $d$. Then the probability that one of the endpoints of a fixed edge is not labeled by $d$ is at most $2\frac{d-1}{n}$. Thus, the expected number of edges in such a solution is $m- 2m\frac{d-1}{n}= m\frac {n-2(d-1)}{n}\geq m\frac{n-2d}{n}$. Thus, the optimum solution has at least $m\frac{n-2d}{n}$ edges. Now, if we set $\gamma$ to $\frac{n-2d}{n\log(d)}$ we have
\begin{align*}
\gamma \log(|Sol_G^k|) f_k \frac m n & \leq \frac{n-2d}{n\log(d)} \log(d^n)  \frac m n 
= \frac{n-2d}{n\log(d)} \log(d)  m 
= m\frac{n-2d}{n}
\leq opt.
\end{align*}
\end{proofof}

%% file: UniformSampling.bbl
\begin{thebibliography}{10}

\bibitem{ahn2012graph}
Kook~Jin Ahn, Sudipto Guha, and Andrew McGregor.
\newblock Graph sketches: sparsification, spanners, and subgraphs.
\newblock In {\em Proceedings of the 31st symposium on Principles of Database
  Systems}, pages 5--14. ACM, 2012.

\bibitem{bahmani2012densest}
Bahman Bahmani, Ravi Kumar, and Sergei Vassilvitskii.
\newblock Densest subgraph in streaming and mapreduce.
\newblock {\em Proceedings of the VLDB Endowment}, 5(5):454--465, 2012.

\bibitem{bhattacharya2015space}
Sayan Bhattacharya, Monika Henzinger, Danupon Nanongkai, and Charalampos~E
  Tsourakakis.
\newblock Space-and time-efficient algorithm for maintaining dense subgraphs on
  one-pass dynamic streams.
\newblock In {\em STOC}, 2015.

\bibitem{charikar2000greedy}
Moses Charikar.
\newblock Greedy approximation algorithms for finding dense components in a
  graph.
\newblock In {\em Approximation Algorithms for Combinatorial Optimization},
  pages 84--95. Springer, 2000.

\bibitem{cdim03}
Graham Cormode, Mayur Datar, Piotr Indyk, and S.~Muthukrishnan.
\newblock Comparing data streams using hamming norms (how to zero in).
\newblock {\em {IEEE} Trans. Knowl. Data Eng.}, 15(3):529--540, 2003.

\bibitem{dourisboure2007extraction}
Yon Dourisboure, Filippo Geraci, and Marco Pellegrini.
\newblock Extraction and classification of dense communities in the web.
\newblock In {\em Proceedings of the 16th international conference on World
  Wide Web}, pages 461--470. ACM, 2007.

\bibitem{esfandiari2014online}
Hossein Esfandiari, MohammadTaghi HajiAghayi, Mohammad~Reza Khani, Vahid
  Liaghat, Hamid Mahini, and Harald R{\"a}cke.
\newblock Online stochastic reordering buffer scheduling.
\newblock In {\em Automata, Languages, and Programming}, pages 465--476.
  Springer, 2014.

\bibitem{feigenblat2011exponential}
Guy Feigenblat, Ely Porat, and Ariel Shiftan.
\newblock Exponential time improvement for min-wise based algorithms.
\newblock In {\em Proceedings of the twenty-second annual ACM-SIAM symposium on
  Discrete Algorithms}, pages 57--66. SIAM, 2011.

\bibitem{gibson2005discovering}
David Gibson, Ravi Kumar, and Andrew Tomkins.
\newblock Discovering large dense subgraphs in massive graphs.
\newblock In {\em Proceedings of the 31st international conference on Very
  large data bases}, pages 721--732. VLDB Endowment, 2005.

\bibitem{glps10}
Anna~C. Gilbert, Yi~Li, Ely Porat, and Martin~J. Strauss.
\newblock Approximate sparse recovery: optimizing time and measurements.
\newblock In {\em Proceedings of the 42nd {ACM} Symposium on Theory of
  Computing, {STOC} 2010, Cambridge, Massachusetts, USA, 5-8 June 2010}, pages
  475--484, 2010.

\bibitem{jst11}
Hossein Jowhari, Mert Saglam, and G{\'{a}}bor Tardos.
\newblock Tight bounds for lp samplers, finding duplicates in streams, and
  related problems.
\newblock In {\em Proceedings of the 30th {ACM} {SIGMOD-SIGACT-SIGART}
  Symposium on Principles of Database Systems, {PODS} 2011, June 12-16, 2011,
  Athens, Greece}, pages 49--58, 2011.

\bibitem{knpw11}
Daniel~M. Kane, Jelani Nelson, Ely Porat, and David~P. Woodruff.
\newblock Fast moment estimation in data streams in optimal space.
\newblock In {\em Proceedings of the 43rd {ACM} Symposium on Theory of
  Computing, {STOC} 2011, San Jose, CA, USA, 6-8 June 2011}, pages 745--754,
  2011.

\bibitem{kapralov2015streaming}
Michael Kapralov, Sanjeev Khanna, and Madhu Sudan.
\newblock Streaming lower bounds for approximating max-cut.
\newblock In {\em Proceedings of the Twenty-Sixth Annual ACM-SIAM Symposium on
  Discrete Algorithms}, pages 1263--1282. SIAM, 2015.

\bibitem{MFCSMcGregor}
Andrew McGregor, David Tench, Sofya Vorotnikova, and Hoa~T. Vu.
\newblock Densest subgraph in dynamic graph streams.
\newblock In {\em Mathematical Foundations of Computer Science 2015}. Springer,
  2015.

\bibitem{panconesi1997randomized}
Alessandro Panconesi and Aravind Srinivasan.
\newblock Randomized distributed edge coloring via an extension of the
  chernoff--hoeffding bounds.
\newblock {\em SIAM Journal on Computing}, 26(2):350--368, 1997.

\bibitem{schmidt1995chernoff}
Jeanette~P Schmidt, Alan Siegel, and Aravind Srinivasan.
\newblock Chernoff-hoeffding bounds for applications with limited independence.
\newblock {\em SIAM Journal on Discrete Mathematics}, 8(2):223--250, 1995.

\end{thebibliography}
